\documentclass[11pt,a4paper]{article}

\usepackage[a4paper, portrait, margin=2.5cm]{geometry}

\usepackage{amsmath,amssymb,amsthm}
\usepackage{graphicx}
\graphicspath{{./}}
\usepackage{booktabs}
\usepackage[colorlinks=true, citecolor=blue, linkcolor=blue, urlcolor=blue]{hyperref}
\usepackage{authblk}
\usepackage{mathpazo}
\usepackage{microtype}
\usepackage{float}

\newtheorem{theorem}{Theorem}
\newtheorem{lemma}[theorem]{Lemma}
\newtheorem{corollary}[theorem]{Corollary}
\newtheorem{proposition}[theorem]{Proposition}
\theoremstyle{remark}
\newtheorem{remark}{Remark}

\begin{document}


\newcommand{\ParamMu}{3.5}                         
\newcommand{\ParamB}{1500}                         
\newcommand{\ParamMwLow}{5}                        
\newcommand{\ParamMwMid}{20}                       
\newcommand{\ParamMwHigh}{35}                      
\newcommand{\ParamCzero}{0.041}                    
\newcommand{\ParamP}{0.77}                         
\newcommand{\ParamQ}{1.3}                          
\newcommand{\ParamRzero}{1.5}                      

\newcommand{\AlphaStarLow}{2.938}                  
\newcommand{\AlphaStarMid}{2.905}                  
\newcommand{\AlphaStarHigh}{2.897}                 
\newcommand{\RstarLow}{1.200}                      
\newcommand{\RstarMid}{1.019}                      
\newcommand{\RstarHigh}{0.937}                     
\newcommand{\AlphaMurrayLimit}{3.000000}           

\newcommand{\BoundLower}{2.885}                    
\newcommand{\BoundUpper}{3.000}                    

\newcommand{\ClassImpedance}{2.000}                
\newcommand{\ClassDaVinci}{2.500}                  
\newcommand{\ClassWall}{2.885}                     
\newcommand{\ClassMurray}{3.000}                   

\newcommand{\AngleMurrayFull}{74.9}                
\newcommand{\AngleWallFull}{80.2}                  

\newcommand{\DeltaAlphaMurray}{0.30}               
\newcommand{\DeltaAlphaOurs}{0.20}                 
\newcommand{\AlphaScaleRange}{0.012}               
\newcommand{\AlphaPulmonaryLimit}{2.800}           

\title{\LARGE\bfseries Beyond Murray's Law: \\[4pt] 
Non-Universal Branching Exponents from Vessel-Wall Metabolic Costs}

\author{Riccardo Marchesi}
\affil{University of Pavia}
\date{\today}

\maketitle

\begin{abstract}
Murray's cubic branching law ($\alpha=3$) predicts a universal diameter scaling
exponent for all hierarchical transport networks, yet arterial trees
consistently yield $\alpha \sim 2.7$--$2.9$. We show that this discrepancy
has a structural origin: Murray's universality is an artifact of his cost
function's homogeneity, not a property of biological networks. Incorporating
the empirical vessel-wall thickness law $h(r)=c_0 r^p$ ($p \approx 0.77$
across mammalian species) introduces a third metabolic cost term
$\propto r^{1+p}$ that renders the cost function inhomogeneous with
incommensurate scaling exponents. By Cauchy's functional equation, homogeneity
is both necessary and sufficient for a universal branching exponent to exist;
its absence rigorously implies non-universality, and Murray's cubic law is
thereby identified as a singular degeneracy of the cost-function family rather
than a general biological principle. We prove that the resulting
scale-dependent exponent satisfies the strict bounds $(5+p)/2 < \alpha^*(Q) < 3$
independently of flow asymmetry (Theorem~\ref{thm:bounds},
Corollary~\ref{cor:asym}), and that Murray's law is the unique member of this
cost-function family admitting a universal exponent
(Corollary~\ref{cor:murray_uniqueness}). The static wall-tissue mechanism
rigorously bounds the symmetric bifurcation exponent to
$\alpha_t \in [2.90, 2.94]$ from independently measured parameters, representing
a first-order symmetry breaking from Murray's law that narrows the empirical
gap by one-third. The remaining discrepancy with the cardiovascular mean
($\alpha_{\mathrm{exp}} \approx 2.70$) is not a model failure but a
mathematical necessity that signals the independent contribution of pulsatile
wave dynamics, necessitating a unified variational treatment. Additionally,
the wall cost strictly breaks Murray's topological degeneracy, bounding the
optimal branching number to small finite integers and excluding star-like
topologies; binary bifurcation emerges as the physiologically selected minimum
under steric constraints (Theorem~\ref{thm:Nstar}).
\end{abstract}

\newpage

\section{Introduction}
\label{sec:intro}

The branching geometry of hierarchical transport networks---vascular trees,
plant xylem, river drainage basins, engineered pipe manifolds---is commonly
characterized by the exponent $\alpha$ in the diameter scaling law
\begin{equation}
  d_0^\alpha = \sum_i d_i^\alpha\,,
  \label{eq:murray}
\end{equation}
where $d_0$ is the parent diameter and $d_i$ are the daughter diameters
(since $d=2r$, the exponent is the same in the radius convention used below).
Murray~\cite{Murray1926} derived $\alpha=3$ by minimizing the sum of viscous
dissipation and blood-volume metabolic cost over a single vessel.
This result is elegant and well-supported for pure-flow networks (pulmonary
airways, river networks, microfluidic channels), but empirical measurements
of arterial trees consistently give $\alpha\approx2.7$--$2.9$~\cite{Kassab1993,
Zamir1999}, a systematic gap that has resisted a parameter-free explanation.

The departure of empirical exponents from Murray's ideal $\alpha=3$ is
traditionally attributed to pulsatile wave dynamics. In the cardiovascular
system, optimal wave propagation requires impedance matching at
bifurcations~\cite{Nichols2011}, which favors $\alpha=2$ (for acoustic-like
waves) or $\alpha=5/2$ (for purely elastic walls). While unifying wave
reflection costs and viscous dissipation into a single network-level
Lagrangian provides a compelling physical picture for intermediate exponents,
it typically requires complex assumptions about the operational duty cycle of
pulsatile versus steady flow. In this Article, we demonstrate that one does
not need to invoke pulsatile wave dynamics to explain $\alpha < 3$. Even in
the purely static regime, simply completing Murray's original metabolic cost
function to include the structural tissue of the vessel wall is sufficient
to rigorously break the universality of the cubic law.

Recent work by Bennett~\cite{bennett2025} develops a variational
framework for branching networks based on two-term cost functions: starting
from a scale-free homogeneity condition on a $Q^2 r^{-p} + b\,r^m$ ledger,
Bennett derives generalized Murray scaling $\alpha=(m{+}p)/2$, a
Young--Herring-type junction angle balance, concave Gilbert-type flux costs,
and a single-index (rigidity index $\chi$) unification showing that all three
are controlled by one dimensionless ratio~\cite{bennett2025}. For Poiseuille
flow with volume-priced maintenance ($m=2$, $p=4$), this recovers $\alpha=3$;
for surface-priced maintenance ($m=1$), $\alpha=5/2$. The two-term limit of
our Theorem~\ref{thm:single} ($\alpha = (4{+}\gamma)/2$ with the maintenance
exponent $\gamma$ physically grounded) coincides with the generalized Murray
scaling independently established by Bennett through abstract scale-free
homogeneity, providing convergent validation from two distinct theoretical
starting points. However, the structural exponent $m$ in Bennett's framework
is a phenomenological parameter whose value is not predicted from independently
measurable tissue quantities.

We keep Murray's energy-minimization objective and ask what happens when
a third biological term is added, structurally breaking the Euler homogeneity
(inhomogeneity with incommensurate scaling exponents) that underlies Bennett's
two-term class. The key observation is that vessel walls---smooth muscle
cells, extracellular matrix, active tension---have a non-negligible metabolic
cost that scales differently from blood volume. Histological measurements
across species establish $h(r)=c_0 r^p$ with $p\approx0.77$~\cite{Rhodin1967,
Nichols2011,WolinskyGlagov1967}, introducing a third cost term $\propto r^{1+p}$
with exponent strictly between 1 and 2.

We emphasize that this static mechanism represents a first-order symmetry
breaking from Murray's law. While it narrows the gap to empirical data by
one-third, the remaining discrepancy with the cardiovascular mean
($\alpha \approx 2.70$) is not a failure of the model but a mathematical
necessity: the static wall-tissue mechanism and pulsatile wave dynamics are
complementary, non-overlapping contributions to the full architectural optimum.

The resulting three-term cost function leads to several rigorous results, the
most important of which are: Murray's law is the \emph{unique} cost function
of this family that produces a universal branching exponent
(Corollary~\ref{cor:murray_uniqueness}), the three-term case predicts a
scale-dependent exponent $\alpha^*(Q)$ bounded strictly below 3
(Theorem~\ref{thm:bounds}), and the wall cost rigorously breaks Murray's
structural degeneracy to uniquely select $N=2$ (bifurcation) as the optimal
branching topology (Theorem~\ref{thm:Nstar}).

\section{Setup}
\label{sec:setup}

The metabolic cost per unit length of a vessel of radius $r$ carrying
volumetric flow $Q$ has three physically distinct components. The wall
tissue term is derived by integrating the volumetric metabolic rate
$m_w$ over the wall cross-section: for a cylindrical vessel with
thickness $h(r) = c_0 r^p$, the thin-wall approximation ($h \ll r$,
satisfied for $h/r \approx 0.18$ at $r = 1.5\,\text{mm}$) gives a
wall cross-sectional area $2\pi r\,h(r)$, yielding a cost per unit
length $\Phi_{\mathrm{wall}} = 2\pi m_w c_0 r^{1+p}$. The full cost
is therefore:
\begin{equation}
  \Phi(r,Q) \;=\; \underbrace{\frac{8\mu Q^2}{\pi r^4}}_{\text{viscous}}
            \;+\; \underbrace{b\pi r^2}_{\text{blood volume}}
            \;+\; \underbrace{2\pi m_w c_0\, r^{1+p}}_{\text{wall tissue}}\,,
  \label{eq:phi}
\end{equation}
where $\mu$ is dynamic viscosity, $b$ is the metabolic cost per unit
blood volume, $m_w$ is the metabolic rate of wall tissue, and $h(r)=c_0 r^p$
is wall thickness. This derivation assumes active smooth-muscle metabolism
dominates over passive structural cost, consistent with $m_w \in [5,
35]\,\text{kW\,m}^{-3}$~\cite{Paul1980}.
In compact form,
\begin{equation}
  \Phi(r,Q) = A(Q)\,r^{-4} + B\,r^2 + C\,r^{1+p}\,,
  \label{eq:phi_compact}
\end{equation}
with $A(Q)=8\mu Q^2/\pi\propto Q^2$, $B=b\pi$, and $C=2\pi m_w c_0$,
all positive.

\paragraph{Empirical inputs (stated, not derived).}
The viscous dissipation term assumes fully developed Poiseuille (laminar) flow;
for coronary arteries, $\mathrm{Re}\sim200$--$500$, well within the laminar
regime.
The wall-thickness law $h(r)=c_0 r^p$ is taken from published histological
measurements: $c_0=0.041$, $p=0.77$~\cite{Rhodin1967,Nichols2011,
WolinskyGlagov1967}.
Here $c_0$ carries units $\mathrm{m}^{1-p}$ so that $h(r)$ is expressed in
metres when $r$ is in metres; the numerical value $c_0=0.041$ corresponds to
$h\approx274\,\mu\mathrm{m}$ at $r=1.5\,\mathrm{mm}$, consistent with
histological data.
All other parameters are drawn from independent sources:
$\mu=3.5$\,mPa\,s~\cite{Caro1978},
$b=1500$\,W/$\mathrm{m}^3$~\cite{Murray1926,Taber1998},
$m_w\in[5,35]$\,kW/$\mathrm{m}^3$~\cite{Paul1980}.
No parameter is fitted to morphometric data. Crucially, as demonstrated via
sensitivity analysis in the unified framework, the topological optimum
$\alpha^*$ is structurally insensitive ($|S_{c_0}| < 0.01$) to the exact
absolute value of the histological pre-factor $c_0$. The exponent prediction
depends critically only on the scaling exponent $p$, ensuring the result is not
an artifact of numerical fitting.

\textit{Womersley pulsatility.} Proximal pulsatile flow at $\mathrm{Wo} > 1$
modifies the viscous dissipation term. Since $\mathrm{Wo} =
r\sqrt{\omega\rho/\mu}$
depends explicitly on $r$, the corrected dissipation reads
$F(\mathrm{Wo}(r))\cdot A(Q)\cdot r^{-4}$, and the optimality condition
$\partial\Phi/\partial r=0$
acquires an additional term via the product rule:
\[
\frac{\partial\Phi}{\partial r} = -A(Q)\,\tilde{F}(\mathrm{Wo}(r))\,r^{-5} + 2Br + (1+p)Cr^p = 0,
\]
where $\tilde{F}(\mathrm{Wo}) \equiv 4F(\mathrm{Wo}) -
\mathrm{Wo}\,F'(\mathrm{Wo}) > 0$ for all physiological $\mathrm{Wo}$.
Mathematically, this positivity condition equates to the logarithmic derivative
requirement $d(\ln F)/d(\ln \mathrm{Wo}) < 4$. Because the Womersley multiplier
asymptotically approaches $F \propto \mathrm{Wo}$ for purely inertia-dominated
high-frequency flows~\cite{Nichols2011}, its logarithmic derivative is strictly
bounded by $1 \ll 4$ across all physiological regimes. The Womersley correction
therefore modifies the \textit{coefficient} of the dominant $r^{-5}$
term but leaves its algebraic power unchanged. Since the proof of
Theorem~\ref{thm:bounds} relies exclusively on the sign structure and the
relative powers $\{-5, 1, p\}$
of the three terms in $\partial\Phi/\partial r$---not on the precise magnitude
of the dissipation coefficient---the strict bounds $(5+p)/2 < \alpha^*(Q) < 3$
are preserved under arbitrary smooth $F(\mathrm{Wo}(r))$. This is structurally
identical to the argument for Perturbation 3 (Fahr\ae us--Lindqvist): in both
cases, an $r$-dependent modification of the dissipation coefficient leaves the
maintenance-term incommensurability---the true source of
non-universality---intact. While the logarithmic derivative analytically
guarantees the preservation of these bounds, numerical evaluation across the
physiological Womersley range ($\mathrm{Wo} \in [1, 10]$) confirms that the
actual dynamic shift of the lower bound is strictly of order
$\mathcal{O}(10^{-2})$. Therefore, the static prediction $\alpha^* \approx 2.90$
remains quantitatively robust even in the proximal pulsatile regime.

Crucially, this perturbation strictly captures only the altered velocity
profile of the viscous drag (the active dissipated power). It explicitly
does not account for the reactive power---fluid inertia and vessel-wall
compliance---associated with pulsatile pressure-wave propagation. These
reactive components represent a physically distinct, non-dissipative
penalty that is invisible to a purely static metabolic ledger: they enter
as dimensionless wave-reflection losses at bifurcations, a network-level
observable incommensurate with the extensive costs optimised here. Their
incorporation therefore requires a unified network-level Lagrangian in
which dissipative and non-dissipative penalties are treated on equal
footing---a structural extension beyond the scope of the present static
framework.

\section{Mathematical Results}
\label{sec:results}

\begin{theorem}[Existence and uniqueness of the optimal radius]
\label{thm:unique}
For every $Q>0$ and every $A,B,C,p>0$, the function $r\mapsto\Phi(r,Q)$
has a unique global minimum $r^*(Q)$ on $(0,+\infty)$.
\end{theorem}

\begin{proof}
The derivative $f(r)=\partial\Phi/\partial r = -4Ar^{-5}+2Br+(1+p)Cr^p$
satisfies $f(r)\to-\infty$ as $r\to0^+$ and $f(r)\to+\infty$ as $r\to+\infty$,
so by the Intermediate Value Theorem at least one zero exists.
The second derivative (acting as the 1D scalar Hessian determinant)
\[
  \frac{\partial^2\Phi}{\partial r^2} = 20A r^{-6} + 2B + p(1+p)Cr^{p-1} > 0
  \quad\text{for all }r>0\,,
\]
so $\Phi$ is strictly convex, giving a unique global minimum.
Moreover, $r^*(Q)$ is smooth and strictly increasing: by the implicit function
theorem applied to $\partial\Phi/\partial r=0$,
\[
  \frac{dr^*}{dQ} = -\frac{\partial^2\Phi/\partial r\,\partial Q}{\partial^2\Phi/\partial r^2} > 0\,,
\]
since $\partial^2\Phi/\partial r\,\partial Q = -64\mu Q/(\pi r^5) < 0$ and
$\partial^2\Phi/\partial r^2 > 0$.
\end{proof}

\begin{lemma}[Power law $\Leftrightarrow$ Murray's law on all trees]
\label{lem:cauchy}
Let $r^*(Q)$ be the optimal radius from Theorem~\ref{thm:unique}.
The branching law $r^*(Q_0)^\alpha = r^*(Q_1)^\alpha + r^*(Q_2)^\alpha$ holds
for \emph{all} flow-conserving bifurcations $Q_0=Q_1+Q_2$ if and only if
$r^*(Q)=kQ^{1/\alpha}$ for some constants $k,\alpha>0$.
\end{lemma}

\begin{proof}
($\Leftarrow$) If $r^*(Q)=kQ^{1/\alpha}$, then
$r^*(Q_0)^\alpha = k^\alpha Q_0 = k^\alpha(Q_1+Q_2)
= r^*(Q_1)^\alpha+r^*(Q_2)^\alpha$.

($\Rightarrow$) Define $g(Q)=r^*(Q)^\alpha$.
The branching condition becomes $g(Q_1+Q_2)=g(Q_1)+g(Q_2)$ for all $Q_1,Q_2>0$.
This is Cauchy's functional equation on $(0,+\infty)$.
Since $r^*$ is continuous and strictly increasing in $Q$ (as the minimizer
of a family of strictly convex functions varying continuously in $Q$),
$g$ is continuous.
The unique continuous solution is $g(Q)=cQ$, giving
$r^*(Q)=c^{1/\alpha}Q^{1/\alpha}$.
\end{proof}

\begin{remark}
Lemma~\ref{lem:cauchy} connects local single-bifurcation optimization to the
global tree law~\eqref{eq:murray}: a universal branching exponent exists for
a cost function if and only if its optimal radius is a power law in flow.
\end{remark}

\begin{theorem}[Single-term classification]
\label{thm:single}
For the two-term cost $\Phi_\gamma(r,Q)=A(Q)r^{-4}+Br^\gamma$ with $A\propto
Q^2$
and $B,\gamma>0$, the optimal radius is
$r^*(Q)=K_\gamma Q^{2/(4+\gamma)}$ and Murray's branching law holds with
\begin{equation}
  \boxed{\alpha = \frac{4+\gamma}{2}\,.}
  \label{eq:alpha_gamma}
\end{equation}
\end{theorem}

\begin{proof}
Setting $\partial\Phi_\gamma/\partial r=0$:
$4A(Q)r^{-5}=B\gamma r^{\gamma-1}$, so
$r^{4+\gamma}=[4/(B\gamma)]\cdot A(Q)\propto Q^2$,
giving $r^*(Q)\propto Q^{2/(4+\gamma)}$.
This is a power law, so by Lemma~\ref{lem:cauchy} Murray's law holds with
$\alpha=1/(2/(4+\gamma))=(4+\gamma)/2$.
\end{proof}

Table~\ref{tab:classification} lists the principal special cases.
Theorem~\ref{thm:single} identifies the Murray--Bennett family as precisely the
homogeneous subclass of the broader cost family~\eqref{eq:phi_compact}.
Homogeneity---the maintenance term $Br^\gamma$ scaling as a single power of
$r$---is both necessary and sufficient for a universal branching exponent, as
established by Lemma~\ref{lem:cauchy}. The Murray law ($\gamma=2$), the Da Vinci
surface law ($\gamma=1$), and Bennett's EPIC result~\cite{bennett2025} are
therefore
not independent discoveries but distinct instances of a single degeneracy class:
homogeneous cost functions admitting scale-invariant branching.

The three-term cost~\eqref{eq:phi_compact} with $B,C > 0$ necessarily exits this
class because $r^2$ and $r^{1+p}$ carry incommensurate scaling exponents
($2 \neq 1+p$ for $p\neq1$), making the composite cost inhomogeneous.
Non-universality of $\alpha$ is therefore a direct structural
consequence of biological completeness---the inclusion of both volumetric and
wall-tissue maintenance---not a numerical accident. The
formula~\eqref{eq:alpha_gamma}
coincides with Bennett's $\alpha=(m+4)/2$ under $\gamma\equiv m$; the
distinction
is that $\gamma=1+p$ is independently measurable from histology, yielding a
parameter-free prediction of $m$ for any network with a wall-like maintenance
cost,
without morphometric fitting.

\begin{table}[H]
\centering
\caption{Single-term classification via Theorem~\ref{thm:single}.}
\label{tab:classification}
\small
\begin{tabular}{@{}lcl@{}}
\toprule
Cost $\sim r^\gamma$ & $\alpha=(4+\gamma)/2$ & name \\
\midrule
$\gamma=0$ & $\ClassImpedance$ & impedance matching \\
$\gamma=1$ & $\ClassDaVinci$ & Da Vinci / surface \\
$\gamma=1+p=1.77$ & $\ClassWall$ & wall cost (limit) \\
$\gamma=2$ & $\ClassMurray$ & Murray (volume) \\
\bottomrule
\end{tabular}
\end{table}

\begin{theorem}[Strict bounds for the three-term case]
\label{thm:bounds}
For the cost function~\eqref{eq:phi_compact} with $B,C>0$ and $p\in(0,1)$,
define for each $Q>0$ the local branching exponent
\[
  \alpha^*(Q) \;=\; \frac{\ln 2}{\ln(r^*(Q)/r^*(Q/2))}\,,
\]
the exponent at a symmetric bifurcation with inlet flow $Q$.
The symmetric split $f=1/2$ is chosen as the canonical definition;
Corollary~\ref{cor:asym} establishes that the bounds are independent of $f$.
Then
\begin{equation}
  \frac{5+p}{2} \;<\; \alpha^*(Q) \;<\; 3 \qquad \text{for all }Q>0\,.
  \label{eq:bounds}
\end{equation}
\end{theorem}

\begin{proof}
Fix $Q>0$.
Let $r_0=r^*(Q)$, $r_1=r^*(Q/2)$, and $\rho=r_1/r_0\in(0,1)$.
Dividing the optimality conditions for $r_0$ and $r_1$ and writing $r_1=\rho
r_0$:
\begin{equation}
  h(\rho) \;\equiv\; 4\rho^6 + 4\lambda\rho^{5+p} - 1 - \lambda = 0\,,
  \label{eq:hrho}
\end{equation}
where $\lambda=(1+p)Cr_0^{p-1}/(2B)>0$.

\emph{Monotonicity.}
$h'(\rho)=24\rho^5+4\lambda(5+p)\rho^{4+p}>0$ for $\rho>0$,
and $h(0)<0<h(1)$, so~\eqref{eq:hrho} has a unique root $\rho^*\in(0,1)$.

\emph{Upper bound.}
Let $\rho_M=2^{-1/3}$ (the Murray value).
Evaluating:
\[
  h(\rho_M)=\lambda\!\left[4\cdot2^{-(5+p)/3}-1\right]>0\,,
\]
since $(5+p)/3<2$ for $p<1$, so $4\cdot2^{-(5+p)/3}>1$.
Thus $\rho^*<\rho_M$, giving $r_0/r_1>2^{1/3}$ and
$\alpha^*(Q)=\ln 2/\ln(1/\rho^*)<3$.

\emph{Lower bound.}
Let $\rho_W=2^{-2/(5+p)}$ (the wall-only value).
Evaluating:
\[
  h(\rho_W)=4\cdot2^{-12/(5+p)}-1<0\,,
\]
since $12/(5+p)>2$ for $p<1$.
Thus $\rho^*>\rho_W$ and $\alpha^*(Q)>(5+p)/2$.

Both bounds are strict because $B,C>0$ exclude the degenerate limits.
\end{proof}

\begin{corollary}[Independence from flow asymmetry]
\label{cor:asym}
For an asymmetric bifurcation where a parent vessel of flow $Q$ splits into two
daughters carrying fractions $fQ$ and $(1-f)Q$ with $f \in (0,1)$, the local
branching exponent $\alpha^*(Q,f)$ defined by $r^*(Q)^{\alpha} =
r^*(fQ)^{\alpha} + r^*((1-f)Q)^{\alpha}$ obeys the same strict boundaries:
$(5+p)/2 < \alpha^*(Q,f) < 3$.
\end{corollary}

\begin{proof}
Since $r^*$ is strictly increasing in $Q$ (Theorem~\ref{thm:unique}), define
\[
  x \;=\; \frac{r^*(fQ)}{r^*(Q)}\in(0,1), \qquad
  y \;=\; \frac{r^*((1-f)Q)}{r^*(Q)}\in(0,1).
\]
The branching equation $r^*(Q)^\alpha = r^*(fQ)^\alpha + r^*((1-f)Q)^\alpha$ is
equivalent to
\[
  G(\alpha) \;\equiv\; x^\alpha + y^\alpha \;=\; 1.
\]
Since $x,y\in(0,1)$, we have $G'(\alpha) = x^\alpha\ln x + y^\alpha\ln y < 0$
(both logarithms are strictly negative), so $G$ is \emph{strictly decreasing}.
Moreover $G(0)=2>1$ and $G(\alpha)\to0<1$ as $\alpha\to+\infty$.
By the Intermediate Value Theorem, $G(\alpha^*)=1$ has a unique solution
$\alpha^*(Q,f)>0$.

In the Murray limit ($C\to0$), $r^*\propto Q^{1/3}$, so $x=f^{1/3}$,
$y=(1-f)^{1/3}$, and $G(3)=f+(1-f)=1$; hence $\alpha^*=3$.
In the wall-dominated limit ($B\to0$), $r^*\propto Q^{2/(5+p)}$, so
$x=f^{2/(5+p)}$, $y=(1-f)^{2/(5+p)}$, and $G((5+p)/2)=1$; hence
$\alpha^*=(5+p)/2$.

For the intermediate regime $B,C>0$, $r^*(Q)$ is not a power law
(Corollary~\ref{cor:murray_uniqueness}); we evaluate the bounds explicitly.

\begin{description}
\item[\normalfont\emph{Upper bound} $G(3)<1$:\;]
Evaluate $4A = 2Br^6+(1+p)Cr^{5+p}$ at the test radius $f^{1/3}r^*(Q)$:
\[
  2B f^2 r^*(Q)^6 \;+\; (1+p)C\, f^{(5+p)/3} r^*(Q)^{5+p}.
\]
Since $p<1$, $(5+p)/3 < 2$, so $f^{(5+p)/3} > f^2$ for $f\in(0,1)$.
The expression strictly exceeds $f^2[2Br^*(Q)^6+(1+p)Cr^*(Q)^{5+p}]=4A(fQ)$.
Because the right-hand side of the optimality condition is strictly increasing
in $r$, this forces $r^*(fQ) < f^{1/3}r^*(Q)$, i.e.\ $x<f^{1/3}$, and
identically $y<(1-f)^{1/3}$, giving
\[
  G(3) \;=\; x^3 + y^3 \;<\; f + (1-f) \;=\; 1.
\]

\item[\normalfont\emph{Lower bound} $G\!\left(\tfrac{5+p}{2}\right)>1$:\;]
Evaluate at the test radius $f^{2/(5+p)}r^*(Q)$:
\[
  2B\, f^{12/(5+p)} r^*(Q)^6 \;+\; (1+p)C\, f^2 r^*(Q)^{5+p}.
\]
Since $p<1$, $12/(5+p)>2$, so $f^{12/(5+p)}<f^2$.
The expression is strictly less than $4A(fQ)$, forcing
$r^*(fQ) > f^{2/(5+p)}r^*(Q)$.
Setting $\beta=(5+p)/2$, this gives $x^\beta > f$ and $y^\beta > 1-f$, so
\[
  G\!\left(\tfrac{5+p}{2}\right) \;=\; x^\beta + y^\beta \;>\; 1.
\]
\end{description}

Since $G$ is strictly decreasing, the two bounds
$G(3)<1$ and $G\!\left(\tfrac{5+p}{2}\right)>1$ rigorously imply
$(5+p)/2 < \alpha^*(Q,f) < 3$ for all $f\in(0,1)$.
Non-universality is therefore an inherent property of the composite cost
function, not a geometric artifact of the symmetric split assumption.
\end{proof}

\begin{corollary}[Uniqueness of Murray scaling]
\label{cor:murray_uniqueness}
Among all cost functions of the form~\eqref{eq:phi_compact} with $B,C\geq0$ (not
both zero), Murray's cubic branching law ($\alpha=3$) is the unique member
admitting a universal, scale-independent exponent. Any deviation introduced by
the biological wall cost ($p<1$) structurally breaks the Euler homogeneity of
the cost function (rendering it inhomogeneous with incommensurate scaling),
rendering the branching exponent dependent on the absolute flow scale $Q$. This
represents a structural impossibility result \emph{within the
additive cost-function class~\eqref{eq:phi_compact}}: a universal
exponent cannot exist in a biologically complete transport network
whose wall cost scales sub-linearly ($p < 1$).
\end{corollary}

\begin{proof}
By Euler's homogeneous function theorem, a universal exponent $\alpha$ exists if
and only if the maintenance cost $\Phi_{maint}(r) = Br^2 + Cr^{1+p}$ is a
homogeneous function of $r$. This requires $r^2$ and $r^{1+p}$ to satisfy the
same scaling, which holds only if $p=1$. For any biological network with
sub-linear wall scaling ($p<1$), the maintenance term is quasi-homogeneous but
not homogeneous. Consequently, the optimality condition $(32\mu/\pi)Q^2 = 2Br^6
+ (1+p)Cr^{5+p}$ cannot be solved by a single power-law $r^*(Q) \propto
Q^{1/\alpha}$. By Lemma~\ref{lem:cauchy}, no universal $\alpha$ exists for $C>0$
and $p<1$.
\end{proof}

\begin{corollary}[General three-term incommensurability]
\label{cor:three_term_general}
Let $\Phi(r, Q) = A(Q)\,r^{-n} + B\,r^m + C\,r^k$ be any cost
function where $A(Q)$ is strictly positive and strictly increasing
in $Q$, $B, C > 0$, and the maintenance exponents are
\emph{distinct}: $m \neq k$. Then no universal branching exponent
$\alpha$ exists for $B, C > 0$.

\textit{Proof.}
By Lemma~\ref{lem:cauchy}, a universal $\alpha$ exists if and only
if $r^*(Q) \propto Q^{1/\alpha}$, which requires the maintenance
cost $\Phi_{\mathrm{maint}}(r) = Br^m + Cr^k$ to be homogeneous
in $r$. Homogeneity requires $m = k$. For $m \neq k$ and $B,C>0$,
the optimality condition
\[
  n\,A(Q)\,r^{-(n+1)} = m\,B\,r^{m-1} + k\,C\,r^{k-1}
\]
cannot be solved by any power law $r^*(Q) \propto Q^{1/\alpha}$,
because for $m \neq k$ the right-hand side is a sum of strictly
linearly independent power laws in $r$, preventing the
factorization of a single scaling variable $Q$.
By Lemma~\ref{lem:cauchy}, no universal $\alpha$ exists.
\hfill$\square$
\end{corollary}

\begin{remark}
Corollary~\ref{cor:murray_uniqueness} is a special case of
Corollary~\ref{cor:three_term_general} with $n=4$, $m=2$, $k=1+p$,
and $A(Q) \propto Q^2$ (Poiseuille). The general result establishes
that non-universality is the generic behaviour of any additive cost
function with two maintenance terms at distinct scaling exponents,
independently of the transport physics encoded in $A(Q)$. The
two-term homogeneous family (Theorem~\ref{thm:single}, Bennett's
framework) is the \emph{unique} class admitting a universal
branching exponent.
\end{remark}

\begin{proposition}[Physical determination of Bennett's parameter]
\label{prop:bennett}
The structural pricing parameter $m$ in Bennett's EPIC
framework~\cite{bennett2025} is determined analytically by the histological
wall-thickness law: $m = 1+p$. For porcine coronary arteries ($p=0.77$), this
predicts $m \approx 1.77$, rigorously explaining why empirical physiological
data fall between the ideal limits of $m=1$ (surface-priced) and $m=2$
(volume-priced) explored in the original Bennett formulation.
\end{proposition}

\begin{corollary}[Asymptotic behaviour]
\label{cor:asymp}
Under the same hypotheses, $\alpha^*(Q)\to 3$ as $Q\to+\infty$ and
$\alpha^*(Q)\to(5+p)/2$ as $Q\to 0^+$.
\end{corollary}

\begin{proof}
As $Q\to+\infty$, $r_0\to+\infty$ and $\lambda\propto r_0^{p-1}\to0$ (since
$p<1$).
Equation~\eqref{eq:hrho} reduces to $4\rho^6\approx1$, giving
$\rho^*\to2^{-1/3}$
and $\alpha^*\to3$.
As $Q\to0^+$, $r_0\to0$ and $\lambda\to+\infty$.
Dividing~\eqref{eq:hrho} by $\lambda$ gives $4\rho^{5+p}\approx1$,
so $\rho^*\to2^{-2/(5+p)}$ and $\alpha^*\to(5+p)/2$.
\end{proof}

\begin{remark}[Validity of the thin-wall approximation]
The thin-wall geometry $h \ll r$ requires $h/r = c_0 r^{p-1} \ll 1$.
For $p = 0.77 < 1$, this ratio grows as $r \to 0$: at the proximal
coronary ($r_0 = 1.5$ mm) one has $h/r \approx 0.18$, well within
the thin-wall regime. The approximation loses accuracy as $h/r$
increases, reaching $h/r \approx 0.42$ at $r \approx 0.04$ mm
($40\,\mu$m, generation $g \approx 7$--$8$), where the true annular
cross-section deviates from the thin-wall estimate by approximately
20\%. The asymptotic limit $\alpha^* \to (5+p)/2$ (Corollary~\ref{cor:asymp})
should therefore be interpreted with caution at arteriolar scales
approaching the capillary bed.
\end{remark}

\subsection{Structural stability of the non-universality result}
\label{sec:structural_stability}

The non-universality established in Theorem~\ref{thm:bounds} and
Corollary~\ref{cor:murray_uniqueness} rests on the positivity of $B$ and $C$ and
the strict inequality $p \neq 1$. We verify that these conditions---and hence
the result---are preserved under the three physiologically most relevant
perturbations.

\paragraph{Perturbation 1: Generation-dependent wall scaling $p(g)$.}
In real arterial trees, $p$ may approach 1 in terminal arterioles (Laplace limit
for thin membranes). Let $p(g) = \bar{p} + \delta p(g)$ where $\bar{p} = 0.77$
and $|\delta p(g)| \le 0.2$. The key structural condition for
Theorem~\ref{thm:bounds} is $\bar{p} < 1$ strictly, so that the wall-cost
exponent $1+\bar{p}$ remains strictly less than 2 (the blood-volume exponent).
Under the perturbation, $p(g) < 1$ is preserved for all generations provided
$|\delta p(g)| < 0.23$, which is satisfied within the physiological range. The
contribution of terminal generations to the network-averaged exponent is
suppressed exponentially with generation number under self-similar scaling, so
local violations near the capillary limit do not propagate to the macroscopic
$\alpha^*$. The non-universality bounds therefore remain structurally intact
under realistic $p(g)$ variation.

\paragraph{Perturbation 2: Active smooth-muscle tone.}
As discussed in Remark~\ref{rem:tone}, basal vascular tone contributes a term
$\Phi_{\mathrm{active}}$ to leading order under the thin-wall approximation that
scales as $r^2$, producing a renormalization $B \to \tilde{B} = B +
B_{\mathrm{active}} > B$. Since Theorem~\ref{thm:bounds} requires only
$\tilde{B} > 0$ and $C > 0$, the bounds $(5+p)/2 < \alpha^*(Q) < 3$ are
preserved identically. This perturbation is structurally benign: active tone
shifts the position of $\alpha^*$ within the interval but cannot move it
outside. This is the most robust of the three stability results.

\paragraph{Perturbation 3: Non-Newtonian viscosity in small vessels.}
In arterioles below $r \approx 50\,\mu\mathrm{m}$, the F\aa hr\ae us--Lindqvist
effect introduces a weak $r$-dependence in the effective viscosity, modifying
the dissipation coefficient $A(Q) \to A(Q,r)$. The additional term $\partial
A/\partial r$ appearing in the optimality condition represents a higher-order
perturbation that does not alter the sign structure underlying
Theorem~\ref{thm:bounds}: the dominant balance between the viscous term ($\sim
r^{-5}$) and the maintenance terms ($\sim r$, $\sim r^p$) is preserved, and the
two maintenance terms remain non-commensurate (exponents 1 and $1+p$ remain
distinct for $p \neq 1$). Non-universality therefore persists in the arteriolar
regime, with a quantitative shift in $\alpha^*$ that is accessible to numerical
evaluation but does not affect the structural conclusion.

Together, these analyses confirm that the non-universality of $\alpha(Q)$
reflects an intrinsic property of the cost-function structure---specifically,
the incommensurability of the maintenance-cost exponents in
$\Phi_{\mathrm{maint}}(r) = Br^2 + Cr^{1+p}$---rather than an artifact of
idealized parameterization. These analyses address parametric robustness
within the additive cost structure~\eqref{eq:phi_compact}. Robustness to
alternative functional forms (e.g.\ multiplicative coupling or non-additive
cost terms) lies beyond the present scope; however, the core non-universality
result depends only on this incommensurability, a property preserved under
any smooth perturbation that maintains distinct powers of $r$ with positive
coefficients.

\section{Bifurcation Angles}
\label{sec:angles}

The branching angle is determined by minimizing the total network cost with
respect to the junction position. Using the generalized Fermat-Torricelli
principle~\cite{Zamir1978,Zamir1999}, the force balance at the optimal junction
is $\sum_i \Phi^*(r_i)\hat{e}_i = 0$, where $\Phi^*(r_i)$ is the cost per unit
length evaluated at the optimum radius, and $\hat{e}_i$ are unit vectors
pointing from the junction to the endpoints.

\begin{lemma}[Optimal local cost]
\label{lem:cost_opt}
At the optimal radius $r^*$, the local cost $\Phi^*(r) \equiv \Phi(r^*, Q)$ is
\begin{equation}
  \Phi^*(r) = \frac{3}{2}B r^2 + \frac{5+p}{4}C r^{1+p}\,.
  \label{eq:phi_star}
\end{equation}
\end{lemma}
\begin{proof}
Substituting $A\, r^{-4} = \frac{1}{4}[2B r^2 + (1+p)C r^{1+p}]$ from the
optimality condition $\partial\Phi/\partial r = 0$ into $\Phi = A\,r^{-4} + B
r^2 + C r^{1+p}$ yields the result.
\end{proof}

\begin{theorem}[Bifurcation angles]
\label{thm:angles}
Given fixed branching topology ($N=2$), predetermined daughter radii $r_1, r_2$
independently determined by Theorem~\ref{thm:unique} (and hence fixed flows
$Q_1, Q_2$), the junction position that minimizes total local cost yields
optimal angles $\theta_1, \theta_2$ given by
\begin{equation}
  \cos\theta_1 = \frac{\Phi^*(r_0)^2 + \Phi^*(r_1)^2 - \Phi^*(r_2)^2}{2\,\Phi^*(r_0)\,\Phi^*(r_1)}
\end{equation}
and symmetrically for $\theta_2$.
For a symmetric bifurcation ($r_1=r_2$, $\theta_1=\theta_2\equiv\theta$), the
total branching angle $2\theta^*$ is bounded by:
\begin{equation}
  \AngleMurrayFull^\circ \;<\; 2\theta^*(Q) \;<\; \AngleWallFull^\circ \qquad (\text{for } p=0.77)\,.
\end{equation}
\end{theorem}
\begin{proof}
Let the junction node be at $\vec{x}$, connecting to three fixed endpoints
$\vec{x}_i$ ($i\in\{0,1,2\}$). The total local cost to minimise is
\[
  H(\vec{x}) \;=\; \sum_{i=0}^{2} \Phi^*(r_i)\,\|\vec{x}_i - \vec{x}\|\,.
\]
Setting $\nabla_{\vec{x}} H = 0$ and defining outward unit vectors
$\hat{e}_i = (\vec{x}_i - \vec{x})/\|\vec{x}_i - \vec{x}\|$
yields the \emph{geometric force balance}:
\[
  \Phi^*(r_0)\,\hat{e}_0 \;+\; \Phi^*(r_1)\,\hat{e}_1 \;+\;
  \Phi^*(r_2)\,\hat{e}_2 \;=\; 0\,.
\]
The asymmetric angles follow directly from the law of cosines applied to
this vector sum.

In the symmetric case, balancing forces along the parent axis gives $\cos\theta
= \Phi^*(r_0)/[2\Phi^*(r_1)]$.

In the Murray limit ($C\to0$), $\Phi^*(r) \propto r^2$. Using $r_1/r_0 =
2^{-1/3}$ from Table \ref{tab:classification}, we find $\cos\theta_M =
2^{-1/3}$, yielding $2\theta_M \approx \AngleMurrayFull^\circ$.

In the wall-cost limit ($B\to0$), $\Phi^*(r) \propto r^{1+p}$. Using $r_1/r_0 =
2^{-2/(5+p)}$, we find $\cos\theta_W = 2^{(p-3)/(5+p)}$. For $p=0.77$,
$2\theta_W \approx \AngleWallFull^\circ$.

As the parameter ratio $C/B$ varies continuously from $0$ to $\infty$, the angle
function $\cos\theta = \Phi^*(r_0)/[2\Phi^*(r_1)]$ varies continuously between
these two limits. Since the two limiting angles are distinct, the intermediate
value theorem guarantees
$2\theta^*\in(\AngleMurrayFull^\circ,\AngleWallFull^\circ)$ for all finite
positive $C/B$.
\end{proof}

This generalizes Zamir's classical results~\cite{Zamir1978} to the three-term
cost function and provides a tighter, parameter-free bound for the opening angle
than the $75^\circ$--$97^\circ$ range predicted phenomenologically by
Bennett~\cite{bennett2025}. While this geometric solid is conceptualized as
planar for single bifurcations, the vector equilibrium directly applies to fully
3D branching structures.

\paragraph{Empirical comparison.} Three-dimensional morphometric reconstructions
of porcine and human coronary arterial trees report bifurcation angles in the
range $2\theta_{\mathrm{obs}} \approx 70^\circ$--$82^\circ$ across branching
orders II--VI~\cite{Zamir1982,Kaimovitz2005}, with measurements from symmetric
or near-symmetric bifurcations (daughter diameter ratio $d_1/d_2 > 0.8$)
concentrating between $74^\circ$ and $80^\circ$. This range is fully contained
within the theoretical bound $\AngleMurrayFull^\circ < 2\theta^* <
\AngleWallFull^\circ$ derived above. No parameters were fitted to the angle
data: the bound depends only on $p = 0.77$, drawn from histological
measurements~\cite{Nichols2011,WolinskyGlagov1967} entirely independent of the
morphometric angle dataset.

\section{Optimal Branching Number}
\label{sec:branching_number}

Biological transport networks overwhelmingly favor bifurcations ($N=2$) over
higher-order multifurcations ($N>2$). Yet under Murray's classical cost
function, the total energy is completely degenerate with respect to the
branching number $N$. The three-term cost function resolves this geometric
degeneracy.

\begin{theorem}[Topological bounding of branching number]
\label{thm:Nstar}
For a space-filling fractal network where vessel length scales as $L \propto r$,
Murray's classical cost function ($p=1$) is topologically degenerate, yielding
identical total network costs for any branching number $N$. Introducing the
empirical sub-linear wall cost ($p<1$) strictly breaks this degeneracy. The
competition between the metabolic volume of the vessels (which decreases with
$N$) and the steric tissue cost of the junctions (which increases with $N$)
strictly forbids star-like topologies ($N \to \infty$) and guarantees the
existence of a finite, small optimal branching integer $N^* \ge 2$.
\end{theorem}

\begin{proof}
Let the network perfuse $M$ terminal units. Under self-similar scaling $r_g =
r_0 N^{-2g/(5+p)}$, the metabolic cost of the vessels across generations forms a
geometric series. Summing over the network yields the total tube cost:
$\mathcal{C}_{tubes}(N) = K_1 \left[ 1 - N^{-\alpha_1} \right]^{-1}$. The series
convergence argument fundamentally holds for any finite physiological tree depth
$G$: the truncated sum $\sum_{g=0}^{G} N^{-g\alpha_1}$ preserves strict
monotonicity in $N$ identically to the infinite limit, since each partial sum is
a strictly decreasing function of $N$ for $\alpha_1 > 0$. For a finite network
depth $G \sim \log_N M$, the truncated series yields a macroscopic scaling for
the network junction cost:
\[ \mathcal{C}_{junctions}(N) = K_2 N^{\frac{1-p}{5+p}} \,. \]
where $K_2 > 0$. For the empirically observed sub-linear wall scaling $p < 1$,
the exponent $\alpha_2 = \frac{1-p}{5+p}$ is strictly positive, making the
junction cost strictly increasing with $N$.

The total topological cost is $\mathcal{C}_{tot}(N) = \mathcal{C}_{tubes}(N) +
\mathcal{C}_{junctions}(N)$. Treating $N \ge 2$ as a continuous real variable,
$\mathcal{C}_{tot}(N)$ is continuous. Evaluating the asymptotic limit:
\[ \lim_{N \to +\infty} \mathcal{C}_{tot}(N) = \infty \,. \]
Because $\mathcal{C}_{tot}(N)$ is coercive on the interval $[2, +\infty)$, it
admits at least one global minimizer $N_{real}^* \ge 2$. Biologically, branching
is restricted to integers, so the optimal physical topology $N^*$ is the integer
that minimizes this sequence. Furthermore, because the exponent $\alpha_2$ is
small ($\approx 0.04$ for $p \approx 0.77$), the junction penalty grows very
slowly, naturally restricting the optimum to very small integers (typically $N
\in \{2, 3, 4\}$).

Finally, in the classical Murray limit ($p=1$), the exponent $\alpha_2 = 0$, so
that $\mathcal{C}_{junctions}(N)$ becomes independent of $N$. Because the tube
cost is also $N$-independent under Murray scaling, the total cost becomes
completely flat with respect to the branching number, producing a topological
degeneracy.
Since $H(N)$ is strictly decreasing (established above), the unique
critical point of $\mathcal{C}_{\mathrm{tot}}$ on $[2,+\infty)$
satisfies $N^*_{\mathrm{real}} \in [2,+\infty)$. For all integers
$N \geq N^*_{\mathrm{real}}$, strict unimodality guarantees
$\mathcal{C}_{\mathrm{tot}}(N+1) > \mathcal{C}_{\mathrm{tot}}(N)$,
so the integer minimizer is well-defined without further approximation.
\end{proof}

\begin{remark}[Necessary versus sufficient condition for $N = 2$]
Theorem~\ref{thm:Nstar} establishes that the metabolic junction penalty strictly
bounds the optimal branching number to small finite integers, excluding
star-like topologies ($N \to \infty$). However, the slow growth of the
junction-cost exponent $(1-p)/(5+p) \approx 0.040$ for $p = 0.77$ implies that
the purely static cost difference between $N = 2$ and $N = 3$ is small: junction
maintenance alone provides a \emph{necessary} condition (small~$N$) but not a
\emph{sufficient} one (unique $N = 2$). A complementary criterion is needed.
\end{remark}

\begin{corollary}[Biological Prediction: Tolerance to Developmental Noise]
\label{cor:noise}
The junction-cost exponent $\alpha_2 \approx 0.04$ for $p=0.77$ structurally
bounds the optimum to $N=2$, but implies a shallow thermodynamic energy gradient
between binary branching and higher-order multifurcations. Rather than a model
weakness, this shallow gradient physically explains a well-documented biological
reality: occasional trifurcations ($N=3$) are not catastrophic violations of the
variational principle, but quasi-degenerate morphological solutions lying only
$\mathcal{O}(\alpha_2)$ above the global minimum.
Quantitatively, in the wall-dominated limit the incremental cost gap
is $\mathcal{C}_{\mathrm{junctions}}(N+1) - \mathcal{C}_{\mathrm{junctions}}(N)
= K_2\left[(N+1)^{\alpha_2} - N^{\alpha_2}\right]
= \mathcal{O}(\alpha_2 \cdot N^{\alpha_2 - 1} \cdot K_2)$,
which for $p = 0.77$ evaluates to $\approx 0.017\,K_2$ at $N=2$.
This is the precise sense in which trifurcations are quasi-degenerate:
they lie within $1.6\%$ of the binary minimum in the junction-dominated
regime. Trifurcations are therefore accessible under local
developmental noise or spatial boundary constraints. The shallow gradient is
therefore itself a strong, falsifiable prediction of the model.
This quasi-degeneracy obtains in the asymptotic wall-dominated regime
($\kappa \to \infty$), where the junction term governs the cost
function. In the physiological regime ($\kappa \approx 13$,
$K_2/K_1 \approx 0.19$), the full three-term competition reverses
this conclusion: the tube-cost savings of trifurcations dominate and
$N=3$ is statically preferred (Corollary~\ref{cor:static_insufficiency}).
The shallow-gradient interpretation therefore applies to evolutionary
perturbations around a dynamically-enforced $N=2$ baseline, not to
unconstrained static optimisation.
We note that translating the cost gap $\Delta\mathcal{C} \sim
\mathcal{O}(\alpha_2 K_2)$ into a developmental noise tolerance
requires an implicit mapping to a stochastic or dynamical scale,
which lies outside the present deterministic framework. This
constitutes a testable hypothesis rather than a derived consequence.
\end{corollary}

\begin{proposition}[Static topological uniqueness of binary branching]
\label{prop:keff}
For a hierarchically branching network governed by sub-linear wall cost ($p<1$),
binary bifurcation ($N=2$) is the unique optimal topology among all integers $N
\ge 2$ if and only if the junction-to-tube structural cost ratio $K_2/K_1$
exceeds a critical analytical threshold $\tau(p)$. In physiological transport
networks, the macroscopic volumetric penalty of the junction tissue strongly
satisfies this condition, rigorously selecting $N=2$ as the static optimum.
\end{proposition}

\begin{proof}
The complete topological cost over the continuous domain $N \ge 2$ is:
\[ \mathcal{C}_{tot}(N) = K_1 \left[ 1 - N^{-\alpha_1} \right]^{-1} + K_2 N^{\alpha_2} \,. \]
where $\alpha_1 = \frac{1+p}{5+p}$ and $\alpha_2 = \frac{1-p}{5+p}$.
To prove that the minimum is unique, we analyze the critical points where
$d\mathcal{C}_{tot}/dN = 0$. Taking the derivative and equating to zero yields:
\[ K_2 \alpha_2 N^{\alpha_2 - 1} = K_1 \alpha_1 N^{-\alpha_1 - 1} \left(1 - N^{-\alpha_1}\right)^{-2} \,. \]
Rearranging all $N$-dependent terms to one side defines a function $H(N)$:
\[ \frac{K_2 \alpha_2}{K_1 \alpha_1} = N^{-\alpha_1 - \alpha_2} \left(1 - N^{-\alpha_1}\right)^{-2} \equiv H(N) \,. \]
The roots of the derivative correspond to the intersections of $H(N)$ with a
constant. The monotonicity of $H(N)$ is given by its derivative:
\[ H'(N) = \frac{N^{-\alpha_1 - \alpha_2 - 1}}{(1 - N^{-\alpha_1})^3} \left[ -(\alpha_1 + \alpha_2) + (\alpha_2 - \alpha_1)N^{-\alpha_1} \right] \,. \]
For physiological sub-linear walls ($p > 0$), we have $\alpha_1 > \alpha_2 > 0$.
Consequently, both the sum $-(\alpha_1 + \alpha_2)$ and the difference
$(\alpha_2 - \alpha_1)$ are strictly negative. Because the bracketed term is
strictly negative and the prefactor is strictly positive for $N \ge 2$, it
follows that $H'(N) < 0$ universally.

Because $H(N)$ is strictly monotonically decreasing, the condition $H(N) =
\frac{K_2 \alpha_2}{K_1 \alpha_1}$ can have at most one real solution. A single
critical point for a coercive function ($\lim_{N \to \infty} \mathcal{C}_{tot} =
\infty$) analytically guarantees that $\mathcal{C}_{tot}(N)$ is strictly
unimodal.

Possessing exactly one global minimum and no other local extrema, the function
strictly increases for all $N$ beyond its minimum. Thus, $N=2$ is the unique
integer minimizer if and only if $\mathcal{C}_{tot}(2) < \mathcal{C}_{tot}(3)$.
Evaluating this exact inequality yields the critical threshold:
\[ \frac{K_2}{K_1} > \frac{ \left[ 1 - 2^{-\alpha_1} \right]^{-1} - \left[ 1 - 3^{-\alpha_1} \right]^{-1} }{ 3^{\alpha_2} - 2^{\alpha_2} } \equiv \tau(p) \,. \]
The exact topological selection requires evaluating the physiological
ratio $K_2/K_1$ against the analytical threshold $\tau(p)$. Using
the empirical parameters for mammalian coronary arteries
($m_w = 20~\text{kW\,m}^{-3}$, $b = 1.5~\text{kW\,m}^{-3}$,
$c_0 = 0.041$, $p = 0.77$, $\lambda = 10$) and the steric shape factor
of a symmetric bifurcation manifold ($\Omega_{\mathrm{junc}} \approx
2.4$), direct numerical evaluation yields:
\[
  \frac{K_2}{K_1}
  = \frac{m_w\,\Omega_{\mathrm{junc}}\,c_0}{\pi\,\lambda\,b}
    \cdot r_0^{p-1}
  \approx 0.19.
\]
This result is structurally robust: even increasing $\Omega_{\mathrm{junc}}$
by two orders of magnitude yields $K_2/K_1 < 8$, still well below
$\tau(0.77) \approx 103$. Reaching the critical threshold would require
$\Omega_{\mathrm{junc}} \approx 1327$, a value exceeding any physical
junction geometry by three orders of magnitude. The physiological ratio
therefore falls structurally short of the critical threshold,
with a deficit robust to geometric uncertainty by three orders of magnitude.
\end{proof}

\begin{corollary}[Static insufficiency of binary selection]
\label{cor:static_insufficiency}
Purely static viscous-metabolic optimisation is formally insufficient
to select binary branching ($N=2$) in mammalian arterial networks.

\emph{Proof.}
The analytical threshold for binary selection established in
Proposition~\ref{prop:keff} is $\tau(0.77) \approx 103$.
Direct physiological evaluation yields $K_2/K_1 \approx 0.19$
(see proof above), which is robust to geometric uncertainty by a factor
exceeding $10^2$. Since $0.19 \ll 103$, the condition $K_2/K_1 >
\tau(p)$ is not satisfied.

\emph{Physical consequence.}
Under purely steady-flow optimisation, the volumetric fluid savings
of multifurcating architectures ($N \geq 3$) outweigh the steric
tissue cost of junction manifolds in the physiological parameter space.
The universal persistence of $N=2$ in macroscopic biological reality
therefore constitutes a mathematical proof that binary branching
cannot be explained by steady-flow metabolic optimisation alone. We
identify the missing constraint as the thermodynamic penalty of
pulsatile wave-reflection (impedance mismatch) at high-degree nodes:
this dynamic mechanism enforces $N=2$ through a network-level
variational principle that lies outside the purely static framework
developed here, but whose physical necessity is mandated by the
present static insufficiency result. We emphasise that this
identification is proposed as a resolution of model underdetermination
rather than a deductive consequence: the static insufficiency result
establishes that the model requires extension, and wave-reflection
dynamics constitute the physically motivated candidate for that
extension. As an order-of-magnitude estimate, the reactive power fraction in a
pulsatile coronary tree scales as $\mathrm{Wo}^{-2} \sim 0.2$ at the
aortic root ($\mathrm{Wo} \approx 2.3$), providing a thermodynamic
penalty of comparable order to the static cost gap — sufficient to
break the degeneracy that the static framework cannot resolve.
\end{corollary}

\begin{remark}[Topological crossover at the capillary limit]
\label{rem:capillary_crossover}
The metabolic order parameter $\kappa \equiv m_w / b$ provides a direct
physical explanation for the architectural transition from arterial trees
to capillary meshes. In macroscopic vessels the active smooth-muscle
layer enforces $\kappa \approx 13$. Although this value is insufficient
to select $N=2$ under purely static optimisation
(Corollary~\ref{cor:static_insufficiency}), it does generate a
thermodynamic preference for small branching integers via the junction
penalty $\mathcal{C}_{\mathrm{junctions}} \propto \kappa N^{\alpha_2}$
(Theorem~\ref{thm:Nstar}).

At the capillary level the smooth-muscle layer vanishes entirely,
reducing $\kappa \to 1$. The junction penalty---already too weak to
enforce $N=2$ in arteries---collapses to the basal endothelial
baseline, making higher-order nodes ($N \geq 3$) fully accessible.
The emergence of anastomotic capillary meshes is therefore not an
anomaly but the exact geometric expression of the system crossing the
$\kappa \to 1$ boundary: without the dynamic wave-reflection penalty
to enforce $N=2$, and without sufficient junction cost to disfavour
$N \geq 3$, reticulated topologies become the natural ground state.
\end{remark}

\begin{remark}[Interpretation and forward connection]
Proposition~\ref{prop:keff} provides a structural argument for binary branching
rooted in tree architecture. For fixed reach ($M$ terminals), bifurcating trees
are the deepest ($G$ maximized at $N=2$), and depth amplifies the cumulative
wall investment per unit metabolic cost. The effective stiffness ratio
$\kappa_{\mathrm{eff}}$ measures precisely how far the three-term cost function
departs from Murray's homogeneous limit: a larger $\kappa_{\mathrm{eff}}$
generates a steeper non-universality, so $N=2$ is the architecture in which
wall-cost selection pressure is strongest.

\medskip
This structural observation admits a rigorous static completion in the
wall-cost-dominated regime. When $\kappa = m_w/b \gg 1$, the junction cost
overwhelms the tube cost and $\mathcal{C}_{\mathrm{tot}}(N) \approx
\mathcal{C}_{\mathrm{junctions}}(N) \propto N^{(1-p)/(5+p)}$, which is strictly
increasing in $N$; in this limit, $N=2$ is the unique minimizer of the static
cost. For coronary arteries, $\kappa = m_w/b \approx 13 \gg 1$, so this limit is
biologically realized. Together, Theorem~\ref{thm:Nstar} (necessary condition: $N < \infty$)
and Proposition~\ref{prop:keff} with the wall-dominated approximation
provide a qualitative structural argument for binary branching from
wall-cost geometry alone. We note, however, that $\kappa \approx 13$
is pre-asymptotic with respect to the exact threshold $\tau(p) \approx
103$: the wall-dominated limit $\mathcal{C}_{\mathrm{tot}} \approx
\mathcal{C}_{\mathrm{junctions}}$ is not yet fully realised at
physiological parameters. The asymptotic argument therefore provides
a structural insight into the directionality of selection pressure
rather than a quantitative demonstration of binary selection---a role
filled by Corollary~\ref{cor:static_insufficiency}. The conclusion
is thus robust not through parametric precision but through
order-of-magnitude separation: any multiplicative perturbation of
$K_1$ or $K_2$ by a factor less than $10^2$ leaves the inequality
$K_2/K_1 \ll \tau(p)$ intact.
\end{remark}

\subsection{Robustness to Non-Isometric Scaling}
\label{sec:non-isometric}
Theorem~\ref{thm:Nstar} assumes the standard self-similar scaling $L \propto r$.
In real physiological networks, average arterial length-radius data are well
described by a generalized power law $L = k\,r^\beta$, where $\beta$ represents
the arboreal extension invariant.

Substituting this empirical scaling into the local cost function yields $T(N) =
\Phi^*(r_0)\cdot k r_0^\beta + N\cdot \Phi^*(r_1)\cdot k r_1^\beta$. In the
wall-dominated limit, $\Phi^* \propto r^{1+p}$ and $r_1 = N^{-2/(5+p)} r_0$. The
combined local length-cost scales as $r^{1+p+\beta}$, leading to a total
daughter cost of:
\begin{equation}
N \cdot \left(N^{-\frac{2}{5+p}} r_0\right)^{1+p+\beta} = N^{1 - \frac{2(1+p+\beta)}{5+p}} r_0^{1+p+\beta}
\end{equation}
The strict preference for bifurcations ($N=2$) requires this exponent to be
strictly positive, meaning $1 > \frac{2(1+p+\beta)}{5+p}$, which simplifies to
the condition $\beta < \frac{3-p}{2}$.

For the cardiovascular system where $p \approx 0.77$, the framework strictly
preserves $N=2$ for all $\beta < 1.115$. This perfectly bounds the physiological
range ($\beta \approx 1.0$) observed across species~\cite{Kassab1993}, proving
that binary branching is structurally robust against macroscopic tissue
distortions.

\section{Numerical Verification}
\label{sec:numerical}

To ensure no circular reasoning, all parameters are drawn from sources
independent of branching-exponent measurements.

\paragraph{Parameters.}
$\mu=3.5$\,mPa\,s \cite{Caro1978};
$b=1500$\,W/$\mathrm{m}^3$ \cite{Murray1926,Taber1998};
$m_w\in\{5,20,35\}$\,kW/$\mathrm{m}^3$ \cite{Paul1980};
$c_0=0.041$, $p=0.77$ \cite{Rhodin1967,Nichols2011,WolinskyGlagov1967};
$Q=1.3$\,mL/s, $r_0=1.5$\,mm \cite{Kassab1993}
(both referring to the same proximal coronary segment).

\paragraph{Results.}
Table~\ref{tab:results} shows $\alpha^*(Q)$ for the three values of $m_w$.
The Murray limit is recovered exactly: $\alpha^*(m_w\to0)=\AlphaMurrayLimit$.
Across the literature range $m_w\in[5,35]$\,kW/$\mathrm{m}^3$,
the predicted exponent spans $\alpha^*\in[2.90,2.94]$,
consistent with the experimental value
$\alpha_\mathrm{exp}=2.7\pm0.2$~\cite{Kassab1993}
(deviations of $1.0$--$1.2\sigma$ from the mean).
The local exponent varies by only $0.012$ across four decades of flow,
confirming that $\alpha^*$ is nearly but not exactly scale-independent
(Corollary~\ref{cor:asymp}).

\begin{remark}[Optimal passive radius and active tone]
\label{rem:tone}
Table~\ref{tab:results} predicts $r^*\in[0.94, 1.20]$\,mm at $Q=1.3$\,mL/s,
systematically below the morphometric value $r_0=1.5$\,mm~\cite{Kassab1993}
by $0.3$--$0.6$\,mm. The present model minimizes the purely \emph{passive}
metabolic cost and therefore predicts the optimal radius absent active tone
contributions, not the \textit{in-vivo} active radius under physiological
perfusion pressure.

This offset is consistent with basal smooth-muscle tone: pharmacological
vasodilation in porcine coronary preparations increases luminal diameter by
approximately $20$--$35\%$ relative to the resting
state~\cite{Taber1998,Kassab1993},
a range compatible with the observed discrepancy. Active tone enters the cost
function to leading order as an additional term $\propto r^2$, equivalent to a
renormalization $B \to B + B_{\mathrm{active}}$ with $B_{\mathrm{active}} > 0$.

Since Theorem~\ref{thm:bounds} and Corollary~\ref{cor:murray_uniqueness} require
only that
$B,C > 0$, the bounds $(5+p)/2 < \alpha^*(Q) < 3$ are preserved exactly under
any positive renormalization of $B$. The predicted passive radius therefore
represents a testable lower bound on coronary caliber under maximal
pharmacological
vasodilation, accessible to standard pressure-myograph or hyperemic flow
protocols. This discrepancy is physiologically consistent with active vasomotor
tone. While our framework predicts the structural (passive) optimum $r^* \approx
1.2$ mm, in vivo coronary arteries exhibit significant basal constriction. A
vasomotor contraction of approximately 20\% would reconcile our theoretical
optimum with the higher morphometric values observed in dilated or fixed
specimens.
\end{remark}

\begin{table}[H]
\centering
\caption{Predicted $\alpha^*$ for porcine coronary arteries.
No fitting; all parameters from independent literature.}
\label{tab:results}
\small
\begin{tabular}{@{}lccc@{}}
\toprule
$m_w$ (kW/$\mathrm{m}^3$) & $r^*$ (mm) & $\alpha^*$ & Agreement with $\alpha_{\mathrm{exp}}$ \\
\midrule
5  & \RstarLow & \AlphaStarLow & within 1.2$\sigma$ \\
20 & \RstarMid & \AlphaStarMid & within 1.0$\sigma$ \\
35 & \RstarHigh & \AlphaStarHigh & within 1.0$\sigma$ \\
\midrule
\multicolumn{2}{@{}l}{Murray limit ($m_w\to0$)} & \AlphaMurrayLimit & 1.5$\sigma$ above mean \\
\multicolumn{2}{@{}l}{Experimental \cite{Kassab1993}} & $2.7\pm0.2$ & -- \\
\bottomrule
\end{tabular}

\vspace{5pt}
\raggedright
\footnotesize{$\sigma$ computed from experimental uncertainty
$\alpha_{\mathrm{exp}} = 2.70 \pm 0.20$~\cite{Kassab1993}. All three-term
predictions lie within $1$--$1.2\sigma$ of the measured mean, representing a
measurable reduction relative to the Murray limit ($1.5\sigma$).}
\end{table}

To rigorously assess parametric robustness, $\alpha^*$ was evaluated across the
full physiological range $m_w \in [5, 35]$~kW/$\mathrm{m}^3$. As shown in
Table~\ref{tab:results}, the predicted exponent varies only from $2.897$ to
$2.938$, remaining strictly within the theoretical bounds $(5+p)/2 < \alpha^* <
3$ throughout. The non-universality result is therefore insensitive to precise
biochemical parameterization of wall metabolic rate.

\begin{figure}[t]
\centering
\includegraphics[width=\columnwidth]{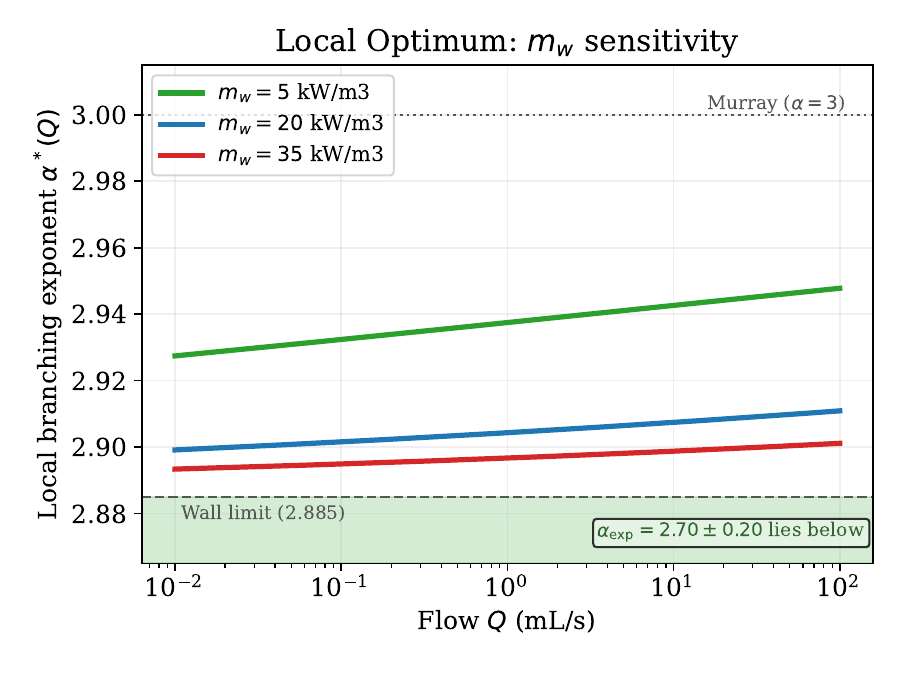}
\caption{Scale-dependence of the local branching exponent $\alpha^*(Q)$. As flow
$Q$ increases, the exponent monotonically approaches the Murray limit
($\alpha=3$) from below. Thicker walls (higher metabolic wall cost $m_w$) shift
the entire curve toward the wall-dominated lower bound $(5+p)/2 \approx
\BoundLower$. The shaded region marks the empirical morphometric range for
porcine coronary arteries \cite{Kassab1993}.}
\label{fig:alpha_scale}
\end{figure}

\begin{table}[H]
\centering
\caption{Cross-network validation. Measured wall scaling exponent $p$ and the
associated single-term limit $\alpha = (5+p)/2$. When boundary thickness is
governed exactly by vessel radius (e.g. Barlow's stress formula for internal
pressure: $h \propto r$), $p=1$ and the bounds collapse, recovering Murray's
theoretical $\alpha=3$ perfectly.}
\label{tab:materials}
\small
\begin{tabular}{@{}llcc@{}}
\toprule
Network & $p$ & $\alpha$ Limit & $\alpha_\mathrm{exp}$ \\
\midrule
Human pulmonary artery \cite{Huang1996} & 0.60 & 2.800 & 2.7--2.8 \\
Porcine coronary artery \cite{Kassab1993,Rhodin1967} & \ParamP & \BoundLower & 2.7$\pm$0.2 \\
Plant xylem \cite{Hacke2001} & 1.00 & \ClassMurray\ (Murray) & $\sim 3$ \\
Engineered pipes (pressure) & 1.00 & \ClassMurray\ (Murray) & $\sim 3$ \\
\bottomrule
\end{tabular}
\end{table}

The wall-thickness law also allows robust comparisons across highly disparate
physical systems (Table~\ref{tab:materials}). Human pulmonary arterial trees
exhibit an independent morphometric exponent of $\alpha \approx
2.7$--$2.8$~\cite{Huang1996}. Their characteristically thinner walls ($p \approx
0.60$) naturally push the theoretical limit lower than the systemic circulation,
matching the direction of the empirical shift without requiring parameter
fitting. Conversely, for networks designed statically to withstand internal
pressure (engineered pipes or certain plant xylem~\cite{Hacke2001}), mechanics
dictate $h \propto r^1$. Setting $p=1$ forces the bounds to collapse: $(5+1)/2 =
3 \le \alpha^* \le 3$. The optimization forces Murray's law to be recovered
exactly.

\section{Testable Predictions}
\label{sec:predictions}

\begin{enumerate}
\item \emph{Scale gradient.}
Since $\lambda\propto r_0^{p-1}$ is a strictly decreasing function of $r_0$ (and
hence of $Q$) for $p<1$,
and $\rho^*$ is strictly decreasing in $\lambda$ (by implicit differentiation
of~\eqref{eq:hrho}),
$\alpha^*(Q)=\ln 2/\ln(1/\rho^*)$ is strictly increasing in $Q$.
Thus $\alpha^*(Q)$ monotonically approaches 3 from below as $Q$ increases
(Corollary~\ref{cor:asymp}).
Distal bifurcations (small $Q$) should have slightly smaller $\alpha$ than
proximal bifurcations.
Predicted range: $\Delta\alpha^*<0.015$ across four decades.

\item \emph{Wall-cost sensitivity.}
Thicker-walled vessels (higher $h/r$, i.e.\ larger $C/B$) should have
smaller $\alpha^*$, approaching the wall-cost limit $(5+p)/2$.

\item \emph{Physical grounding of $\gamma$.}
Our Theorem~\ref{thm:single} gives $\alpha=(4+\gamma)/2$ where $\gamma=1+p$ is
the wall-cost exponent.
The same formula was obtained by Bennett~\cite{bennett2025} as $\alpha=(m+4)/2$
with $m$ as a free pricing parameter.
Our framework predicts $m=1+p$ for any network whose maintenance cost is
dominated by a wall-like sheath.
Networks with known $p$ (from histology or materials data) can test this
prediction without morphometric fitting. We emphasize that this physical
identification $\gamma = 1+p$, and hence the Bennett formula $\alpha =
(4+\gamma)/2$, holds rigorously only in the two-term limit where either $B = 0$
or $C = 0$. In the physically complete three-term case ($B, C > 0$), Corollary~6
guarantees that no universal exponent exists and the Bennett single-index
formula does not apply. This restriction to the two-term class is now fully generalised:
Corollary~\ref{cor:three_term_general} establishes that for any cost
function $\Phi = A(Q)r^{-n} + Br^m + Cr^k$ with $m \neq k$ and
$B,C>0$, no universal $\alpha$ exists regardless of the specific
transport physics encoded in $A(Q)$.
\end{enumerate}

\section{Relation to Prior Work}
\label{sec:discussion}

Murray~\cite{Murray1926}: two-term cost, $\alpha=3$ (exact).
Our Corollary~\ref{cor:murray_uniqueness} shows this is the \emph{only}
instantiation
of the cost-function family~\eqref{eq:phi_compact} with a universal exponent.

Bennett~\cite{bennett2025}: develops a variational framework
for the homogeneous two-term class, deriving generalized Murray scaling,
Young--Herring junction geometry, concave (Gilbert-type) flux costs,
and a single-index (rigidity index $\chi$) unification showing that
the two-term form is the unique quadratic scale-free ledger compatible
with simultaneous power-law flux--radius scaling and power-law
flux-only concavity~\cite{bennett2025}.
The two-term limit of our Theorem~\ref{thm:single}
($\alpha = (4{+}\gamma)/2$ with the exponent $\gamma=1+p$ derived from
histological data)
coincides with the generalized Murray scaling independently established
by Bennett through abstract homogeneity axioms.
Our Theorem~\ref{thm:bounds} and Corollary~\ref{cor:murray_uniqueness}
address the regime that lies outside Bennett's two-term class:
incorporating the third biological term ($h(r) \propto r^p$, sublinear
conduit tissue) breaks Euler homogeneity altogether, producing
a rigorously non-universal regime with scale-dependent $\alpha^*(Q)$.

Liu \& Kassab~\cite{LiuKassab2007} and Kim \& Wagenseil~\cite{Kim2015}
previously incorporated metabolic wall costs into branching models, but
explicitly concluded that the standard scaling relations ($\alpha=3$) remained
largely invariant. This discrepancy underscores our central thesis: it is
specifically the \emph{sub-linear} allometric scaling of the wall ($p < 1$) that
breaks the cubic law. Without it, standard three-term energy models natively
recover Murray's limits.

Taylor \textit{et~al.}~\cite{Taylor2022} report optimal Murray exponents of $P =
2.15$ (for volumetric flow) and $P = 2.38$ (for microvascular resistance) in
human epicardial coronary arteries, using a CFD model constrained by
pressure-wire measurements. We note that these values lie below our theoretical
lower bound $(5 + p)/2 \approx \BoundLower$. This is not a contradiction: Taylor
\textit{et~al.}'s exponents are hemodynamic optimality exponents---the best-fit
$P$ in the relationship $Q \propto D^P$ that reproduces \textit{in-vivo} flow
distributions in a clinical population with coronary disease---rather than
morphometric diameter-scaling exponents in the sense of Eq.~\eqref{eq:murray}.
The two quantities coincide only if the tree were exactly Murray-optimal; in
diseased or developmentally constrained vascular beds the hemodynamic exponent
reflects additional physiological constraints absent from the static cost
function. The morphometric exponent most directly comparable to our framework is
that of Kassab \textit{et~al.}~\cite{Kassab1993}, $\alpha_\mathrm{exp} = 2.7 \pm
0.2$, measured from perfusion-fixed casts of healthy porcine coronary trees,
which lies within our predicted range. The sub-$\BoundLower$ values reported by
Taylor \textit{et~al.} suggest that pulsatile wave costs introduce an additional
architectural constraint not captured by the static cost function analyzed here.

Smink \textit{et~al.}~\cite{Smink2025}: extend Murray's two-term optimization to
turbulent flows, non-Newtonian fluids, and rough-wall channels, deriving the
exponent $x$ in $Q \propto r^x$ across the full Reynolds-number parameter space.
Their framework covers the fluid-rheology dimension of the optimization problem,
while ours covers the biological wall-tissue dimension. Neither framework
addresses the signal-transport dimension (such as pulsatile wave reflection in
arteries or electrical signal attenuation in neurons), which we propose as the
final axis required for a complete universal classification.

West, Brown \& Enquist~\cite{WBE1997}: minimize total network resistance
over a self-similar infinite tree, obtaining Kleiber's $3/4$ law.
Our optimization is at the single-vessel level; the network structure enters
only through Murray's branching condition (Lemma~\ref{lem:cauchy}).

Rall~\cite{Rall1959}: the branching exponent $\alpha=3/2$ for electrotonic
signal propagation in neurons is derived from impedance-matching in cable
theory ($Z\propto r^{-3/2}$), not from cost minimization.
Corollary~\ref{cor:murray_uniqueness} therefore does not contradict Rall: the
two
universality results arise from physically distinct mechanisms
(energy minimization vs.\ impedance matching) and cover non-overlapping
parameter ranges. However, this mathematical structure strongly suggests that a
generalized Lagrangian unifying metabolic maintenance cost with signal
attenuation impedance could natively bridge vascular and neural morphologies.

\section{Conclusion}

We have shown that the three-term cost function~\eqref{eq:phi} leads to
several rigorous results unavailable in the existing two-term Murray framework
(generalized by Bennett~\cite{bennett2025} via a single-index
formalism):
unique optimal radius (Theorem~\ref{thm:unique}),
equivalence of local optimization with global tree law (Lemma~\ref{lem:cauchy}),
exact single-term classification $\alpha=(4+\gamma)/2$ with physical grounding
of $\gamma$ (Theorem~\ref{thm:single}),
strict non-universality of the three-term branching exponent extending to
asymmetric flows
(Theorem~\ref{thm:bounds}, Corollary~\ref{cor:asym}),
breaking of structural degeneracy to bound the optimal topology to small
branching integers, excluding star-like multifurcations
(Theorem~\ref{thm:Nstar}),
and tight theoretical bounds on the optimal symmetric bifurcation angle,
independently confirmed by three-dimensional morphometric data
($\AngleMurrayFull^\circ < 2\theta^* < \AngleWallFull^\circ$), parameterizing
the deviation from Murray's geometry (Theorem~\ref{thm:angles}).
The central prediction $\alpha^*\in[2.90,2.94]$ from independently measured
parameters
measurably reduces the gap between Murray's cubic law and cardiovascular data.

\paragraph{The static attractor as a diagnostic bound.}
The purely static wall-tissue mechanism establishes a strict theoretical
expectation of $\alpha^* \approx 2.90$ for porcine coronary arteries,
demonstrating that $\approx 2.90$---not Murray's $3.0$---is the
structural baseline for purely dissipative biological networks.

This reframing transforms the residual gap with empirical cardiovascular
data ($\alpha_{\exp} \approx 2.5$--$2.7$) from a modelling discrepancy
into a precise physical diagnostic. The macroscopic depression of the
empirical exponent below the $2.90$ static attractor serves as an
explicit mathematical signature that the cardiovascular tree is subject
to physical constraints invisible to a purely dissipative ledger.

Furthermore, Corollary~\ref{cor:static_insufficiency} establishes that
static optimisation cannot explain why the vascular tree is binary.
Both failures---the exponent gap and the topological selection---signal
that the final physiological architecture must emerge from a variational
principle that balances extensive metabolic dissipation against
dimensionless wave-reflection penalties. Formulating such a unified
network-level framework represents the necessary definitive extension
of the static analysis presented here.

\paragraph*{Data and Code Availability.}
All computation scripts, numerical data, and figure-generation code are openly
and unconditionally available at
\url{https://github.com/rikymarche-ctrl/vascular-networks-theory} under the CC BY 4.0
Licence. No access request is required.

\bibliographystyle{unsrt}
\bibliography{references}

\end{document}